\documentclass[proceedings]{stacs}
\stacsheading{2009}{613--624}{Freiburg}
\firstpageno{613}

\usepackage{times}
\usepackage{amsthm}
\usepackage{amsmath}
\usepackage{amsfonts}
\usepackage{latexsym}
\usepackage{algorithm}
\usepackage[noend]{algorithmic}

\newcommand{\field}[1]{\ensuremath{\mathbb{F}_{#1}}}
\newcommand{\Oh}{\ensuremath{\mathcal{O}}}
\newcommand{\eps}{\ensuremath{\varepsilon}}
\newcommand{\expect}{\ensuremath{\mathbb{E}}}
\newcommand{\pr}{\ensuremath{\mathbb{P}}}

\newcommand{\colrel}[1]{\ensuremath{\vartriangleleft_{#1}}}
\newcommand{\badcol}[2]{\ensuremath{\mathrm{Bad}_{#1}(#2)}}

\begin{document}

\title[Local Multicoloring Algorithms]{Local Multicoloring Algorithms:\\
  \small\it Computing a Nearly-Optimal TDMA Schedule in Constant Time}
\author[tds]{F. Kuhn}{Fabian Kuhn}
\address[tds]{MIT, Computer Science and Artificial Intelligence Lab
\newline 32 Vassar St, Cambridge, MA 02139, USA}
\email{fkuhn@csail.mit.edu}

\thanks{For space reasons, most proofs are omitted from this extended
  abstract. A full version can be received from the author's web site
  at
  \texttt{http://people.csail.mit.edu/fkuhn/publications/multicoloring.pdf}.}


\begin{abstract}
  We are given a set $V$ of autonomous agents (e.g.\ the computers of
  a distributed system) that are connected to each other by a graph
  $G=(V,E)$ (e.g.\ by a communication network connecting the agents).
  Assume that all agents have a unique ID between $1$ and $N$ for a
  parameter $N\ge|V|$ and that each agent knows its ID as well as the
  IDs of its neighbors in $G$. Based on this limited information,
  every agent $v$ must autonomously compute a set of colors
  $S_v\subseteq C$ such that the color sets $S_u$ and $S_v$ of
  adjacent agents $u$ and $v$ are disjoint. We prove that there is a
  deterministic algorithm that uses a total of
  $|C|=\Oh(\Delta^2\log(N)/\eps^2)$ colors such that for every node
  $v$ of $G$ (i.e., for every agent), we have $|S_v|\ge
  |C|\cdot(1-\eps)/(\delta_v+1)$, where $\delta_v$ is the degree of
  $v$ and where $\Delta$ is the maximum degree of $G$. For
  $N=\Omega(\Delta^2\log\Delta)$, $\Omega(\Delta^2+\log\log N)$ colors
  are necessary even to assign at least one color to every node (i.e.,
  to compute a standard vertex coloring). Using randomization, it is
  possible to assign an $(1-\eps)/(\delta+1)$-fraction of all colors
  to every node of degree $\delta$ using only
  $\Oh(\Delta\log|V|/\eps^2)$ colors w.h.p.  We show that this is
  asymptotically almost optimal. For graphs with maximum degree
  $\Delta=\Omega(\log|V|)$, $\Omega(\Delta\log|V|/\log\log|V|)$ colors
  are needed in expectation, even to compute a valid coloring.

  The described multicoloring problem has direct applications in the
  context of wireless ad hoc and sensor networks. In order to
  coordinate the access to the shared wireless medium, the nodes of
  such a network need to employ some medium access control (MAC)
  protocol. Typical MAC protocols control the access to the shared
  channel by time (TDMA), frequency (FDMA), or code division multiple
  access (CDMA) schemes. Many channel access schemes assign a fixed
  set of time slots, frequencies, or (orthogonal) codes to the nodes
  of a network such that nodes that interfere with each other receive
  disjoint sets of time slots, frequencies, or code sets. Finding a
  valid assignment of time slots, frequencies, or codes hence directly
  corresponds to computing a multicoloring of a graph $G$. The
  scarcity of bandwidth, energy, and computing resources in ad hoc and
  sensor networks, as well as the often highly dynamic nature of these
  networks require that the multicoloring can be computed based on as
  little and as local information as possible.
\end{abstract}

\maketitle

\vspace*{-5mm}

\section{Introduction}
\label{sec:intro}

In this paper, we look at a variant of the standard vertex coloring
problem that we name graph \emph{multicoloring}. Given an $n$-node
graph $G=(V,E)$, the goal is to assign a set $S_v$ of colors to each
node $v\in V$ such that the color sets $S_u$ and $S_v$ of two adjacent
nodes $u\in V$ and $v\in V$ are disjoint while at the same time, the
fraction of colors assigned to each node is as large as possible and
the total number of colors used is as small as possible. In
particular, we look at the following \emph{distributed} variant of
this multicoloring problem. Each node has a unique identifier (ID)
between $1$ and $N$ for an integer parameter $N\ge n$. The nodes are
\emph{autonomous agents} and we assume that every agent has only very
limited, \emph{local} information about $G$. Specifically, we assume
that every node $v\in V$ merely knows its own ID as well as the IDs of
all its neighbors. Based on this local information, every node $v$
needs to compute a color set $S_v$ such that the color sets computed
by adjacent nodes are disjoint. Since our locality condition implies
that every node is allowed to communicate with each neighbor only
once, we call such a a distributed algorithm a \emph{one-shot
  algorithm}.

We prove nearly tight upper and lower bounds for deterministic and
randomized algorithms solving the above distributed multicoloring
problem. Let $\Delta$ be the largest degree of $G$. We show that for
every $\eps\in(0,1)$, there is a deterministic multicoloring algorithm
that uses $\Oh(\Delta^2\log(N)/\eps^2)$ colors and assigns a
$(1-\eps)/(\delta+1)$-fraction of all colors to each node of degree
$\delta$. Note that because a node $v$ of degree $\delta$ does not
know anything about the topology of $G$ (except that itself has
$\delta$ neighbors), no one-shot multicoloring algorithm can assign
more than a $1/(\delta+1)$-fraction of the colors to all nodes of
degree $\delta$ (the nodes could be in a clique of size $\delta+1$).
The upper bound proof is based on the probabilistic method and thus
only establishes the existence of an algorithm. We describe an
algebraic construction yielding an explicit algorithm that achieves
the same bounds up to polylogarithmic factors.  Using
$\Oh(\Delta^2\log^2N)$ colors, for a value $\eps>0$, the algorithm
assigns a $\eps/\Oh(\delta^{1+\eps}\log N)$-fraction of all colors to
nodes of degree $\delta$. At the cost of using
$\Oh(\Delta^{\log^*N}\log N)$ colors, it is even possible to improve
the fraction of colors assigned to each node by a factor of $\log N$.
The deterministic upper bound results are complemented by a lower
bound showing that if $N=\Omega(\Delta^2\log\Delta)$, even for the
standard vertex coloring problem, every deterministic one-shot
algorithm needs to use at least $\Omega(\Delta^2+\log\log N)$ colors.

If we allow the nodes to use randomization (and only require that the
claimed bounds are obtained with high probability), we can do
significantly better. In a randomized one-shot algorithm, we assume
that every node can compute a sequence of random bits at the beginning
of an algorithm and that nodes also know their own random bits as well
as the random bits of the neighbors when computing the color set. We
show that for $\eps\in (0,1)$, with high probability,
$\Oh(\Delta\log(n)/\eps^2)$ colors suffice to assign a
$(1-\eps)/(\delta+1)$-fraction of all colors to every node of degree
$\delta$. If $\log n\le\Delta\le n^{1-\eps}$ for a constant $\eps>0$,
we show that every randomized one-shot algorithm needs at least
$\Omega(\Delta\log n/\log\log n)$ colors. Again, the lower bound even
holds for standard vertex coloring algorithms where every node only
needs to choose a single color.

Synchronizing the access to a common resource is a typical application
of coloring in networks. If we have a $c$-coloring of the network
graph, we can partition the resource (and/or time) into $c$ parts and
assign a part to each node $v$ depending on $v$'s color. In such a
setting, it seems natural to use a multicoloring instead of a standard
vertex coloring and assign more than one part of the resource to every
node. This allows to use the resource more often and thus more
efficiently.

The most prominent specific example of this basic approach occurs in
the context of media access control (MAC) protocols for wireless ad
hoc and sensor networks. These networks consist of autonomous wireless
devices that communicate with each other by the use of radio signals.
If two or more close-by nodes transmit radio signals at the same time,
a receiving node only hears the superposition of all transmitted
signals. Hence, simultaneous transmissions of close-by nodes interfere
with each other and we thus have to control the access to the wireless
channel. A standard way to avoid interference between close-by
transmissions is to use a time (TDMA), frequency (FDMA), or code
division multiple access (CDMA) scheme to divide the channel among the
nodes. A TDMA protocol divides the time into time slots and assigns
different time slots to conflicting nodes. When using FDMA, nodes that
can interfere with each other are assigned different frequencies,
whereas a CDMA scheme uses different (orthogonal) codes for
interfering nodes. Classically, TDMA, FDMA, and CDMA protocols are
implemented by a standard vertex coloring of the graph induced by the
interference relations. In all three cases, it would be natural to use
the more general multicoloring problem in order to achieve a more
effective use of the wireless medium.  Efficient TDMA schedules, FDMA
frequency assignments, or CDMA code assignments are all directly
obtained from a multicoloring of the interference graph where the
fraction of colors assigned to each nodes is as large as possible. It
is also natural to require that the total number of colors is small.
This keeps the length of a TDMA schedule or the total number of
frequencies or codes small and thus helps to improve the efficiency
and reduce unnecessary overhead of the resulting MAC protocols.

In contrast to many wired networks, wireless ad hoc and sensor
networks typically consist of small devices that have limited
computing and storage capabilities. Because these devices operate on
batteries, wireless nodes also have to keep the amount of computation
and especially communication to a minimum in order to save energy and
thus increase their lifetime. As the nodes of an ad hoc or sensor
network need to operate without central control, everything that is
computed, has to be computed by a distributed algorithm by the nodes
themselves.  Coordination between the nodes is achieved by exchanging
messages. Because of the resource constraints, these distributed
algorithms need to be as simple and efficient as possible.  The
messages transmitted and received by each node should be as few and as
short as possible. Note that because of interference, the bandwidth of
each local region is extremely limited. Typically, for a node $v$, the
time needed to even receive a single message from all neighbors is
proportional to the degree of $v$ (see e.g.\ \cite{radiocoloring}). As
long as the information provided to each node is symmetric, it is
clear that every node needs to know the IDs of all adjacent nodes in
$G$ in order to compute a reasonably good multicoloring of $G$. Hence,
the one-shot multicoloring algorithms considered in this paper base
their computations on the minimum information needed to compute a
non-trivial solution to the problem. Based on the above observations,
even learning the IDs of all neighbors requires quite a bit of time
and resources. Hence, acquiring significantly more information might
already render an algorithm inapplicable in practice.\footnote{It
  seems that in order to achieve a significant improvement on the
  multicolorings computed by the algorithms presented in this paper,
  every node would need much more information. Even if every node
  knows its complete $O(\log\Delta)$-neighborhood, the best
  deterministic coloring algorithm that we are aware of needs
  $\Theta(\Delta^2)$ colors.}

As a result of the scarcity of resources, the size and simplicity of
the wireless devices used in sensor networks, and the dependency of
the characteristic of radio transmissions on environmental conditions,
ad hoc and sensor networks are much less stable than usual wired
networks. As a consequence, the topology of these networks (and of
their interference graph) can be highly dynamic. This is especially
true for ad hoc networks, where it is often even assumed that the
nodes are mobile and thus can move in space. In order to adapt to such
dynamic conditions, a multicoloring needs to be recomputed
periodically.  This makes the resource and time efficiency of the used
algorithms even more important. This is particularly true for the
locality of the algorithms. If the computation of every node only
depends on the topology of a close-by neighborhood, dynamic changes
also only affect near-by nodes.

The remainder of the paper is organized as follows. In Section
\ref{sec:relwork}, we discuss related work. The problem is formally
defined in Section \ref{sec:problem}. We present the deterministic and
randomized upper bounds in Section \ref{sec:upperbounds} and the lower
bounds in Section \ref{sec:lowerbounds}.

\vspace*{-2mm}
\section{Related Work}
\label{sec:relwork}

There is a rich literature on distributed algorithms to compute
classical vertex colorings (see e.g.
\cite{awerbuch89,cole86,goldberg88,oneround,linial92,panconesi95}).
The paper most related to the present one is \cite{oneround}. In
\cite{oneround}, deterministic algorithms for the standard coloring
problem in the same distributed setting are studied (i.e., every node
has to compute its color based on its ID and the IDs of its
neighbors). The main result is a $\Omega(\Delta^2/\log^2\Delta)$ lower
bound on the number of colors.  
The first paper to study distributed coloring is a seminal paper by
Linial \cite{linial92}. The main result of \cite{linial92} is an
$\Omega(\log^*n)$-time lower bound for coloring a ring with a constant
number of colors. As a corollary of this lower bound, one obtains an
$\Omega(\log\log N)$ lower bound on the number of colors for
deterministic one-shot coloring algorithms as studied in this paper.
Linial also looks at distributed coloring algorithms for general graph
and shows that one can compute an $\Oh(\Delta^2)$-coloring in time
$\Oh(\log^*n)$. In order to color a general graph with less colors,
the best known distributed algorithms are significantly
slower.\footnote{In \cite{demarco01}, it is claimed that an
  $\Oh(\Delta)$ coloring can be computed in time
  $\Oh(\log^*(n/\Delta))$. However, the argumentation in
  \cite{demarco01} has a fundamental flaw that cannot be fixed
  \cite{pelcpersonal}.} Using randomization, an $\Oh(\Delta)$-coloring
can be obtained in time $\Oh(\sqrt{\log n})$ \cite{kothapalli06}.
Further, the fastest algorithm to obtain a $(\Delta+1)$-coloring is
based on an algorithm to compute a maximal independent set by Luby
\cite{luby86} and on a reduction described in \cite{linial92} and has
time complexity $\Oh(\log n)$. The best known deterministic algorithms
to compute a $(\Delta+1)$-coloring have time complexities
$2^{\Oh(\sqrt{\log n})}$ and $\Oh(\Delta\log\Delta+\log^*n)$ and are
described in \cite{panconesi95} and \cite{oneround}, respectively.
For special graph classes, there are more efficient deterministic
algorithms. It has long been known that in rings \cite{cole86} and
bounded degree graphs \cite{goldberg88,linial92}, a
$(\Delta+1)$-coloring can be computed in time $\Oh(\log^*n)$. Very
recently, it has been shown that this also holds for the much larger
class of graphs with bounded local independent sets
\cite{schneider08}. In particular, this graph class contains all graph
classes that are typically used to model wireless ad hoc and sensor
networks. Another recent result shows that graphs of bounded
arboricity can be colored with a constant number of colors in time
$\Oh(\log n)$ \cite{barenboim08}.

Closely related to vertex coloring algorithms are distributed
algorithms to compute edge colorings
\cite{czygrinow01,grable97,panconesi97}. In a seminal paper, Naor and
Stockmeyer were the first to look at distributed algorithms where all
nodes have to base their decisions on constant neighborhoods
\cite{naor93}. It is shown that a weak coloring with $f(\Delta)$
colors (every node needs to have a neighbor with a different color)
can be computed in time $2$ if every vertex has an odd degree. Another
interesting approach is taken in \cite{advice} where the complexity of
distributed coloring is studied in case there is an oracle that gives
some nodes a few bits of extra information.

There are many papers that propose to use some graph coloring variant
in order to compute TDMA schedules and FDMA frequency or CDMA code
assignments (see e.g.\
\cite{telecommunications,gandham05,herman04,mecke07,ramanathan99,rhee06,zhang08}).
Many of these papers compute a vertex coloring of the network graph
such that nodes at distance at most $2$ have different colors. This
guarantees that no two neighbors of a node use the same time slot,
frequency, or code. Some of the papers also propose to construct a
TDMA schedule by computing an edge coloring and using different time
slots for different edges. Clearly, it is straight-forward to use our
algorithms for edge colorings, i.e., to compute a multicoloring of the
line graph. With the exception of \cite{herman04} all these papers
compute a coloring and assign only one time slot, frequency, or code
to every node or edge. In \cite{herman04}, first, a standard coloring
is computed. Based on this coloring, an improved slot assignment is
constructed such that in the end, the number of slots assigned to a
node is inversely proportional to the number of colors in its
neighborhood.  


\vspace*{-3mm}
\section{Formal Problem Description}
\label{sec:problem}

\vspace*{-3mm}
\subsection{Mathematical Preliminaries}
\label{sec:prelim}

Throughout the paper, we use $\log(\cdot)$ to denote logarithms to
base $2$ and $\ln(\cdot)$ to denote natural logarithms, respectively.
By $\log^{(i)}x$ and by $\ln^{(i)}x$, we denote the $i$-fold
applications of the logarithm functions $\log$ and $\ln$ to $x$,
respectively\footnote{We have $\log^{(0)}x=\ln^{(0)}x=x$,
  $\log^{(i+1)}x=\log(\log^{(i)}x)$, and
  $\ln^{(i+1)}x=\ln(\ln^{(i)}x)$. Note that we also use $\log^ix=(\log
  x)^i$ and $\ln^ix=(\ln x)^i$}. The log star function is defined as
$\log^*n:=\min_i \{\log^{(i)}n\le1\}$. We also use the following
standard notations. For an integer $n\ge1$, $[n]=\{1,\ldots,n\}$. For
a finite set $\Omega$ and an integer $k\in\{0,\ldots,|\Omega|\}$,
${\Omega\choose k}=\{S\in2^{\Omega}:|S|=k\}$. The term with high
probability (w.h.p.) means with probability at least $1-1/n^c$ for a
constant $c\ge1$.

\subsection{Multicoloring}

The multicoloring problem that was introduced in Section
\ref{sec:intro} can be formally defined as follows.

\begin{definition}[Multicoloring]
  An $(\rho(\delta),k)$-multicoloring $\gamma$ of a graph $G=(V,E)$ is
  a mapping $\gamma:V\to2^{[k]}$ that assigns a set
  $\gamma(v)\subset[k]$ of colors to each node $v$ of $G$ such that
  $\forall\{u,v\}\in E:\gamma(u)\cap\gamma(v)=\emptyset$ and such that
  for every node $v\in V$ of degree $\delta$,
  $|\gamma(v)|/k\ge\rho(\delta)/(\delta+1)$.
\end{definition}

We call $\rho(\delta)$ the \emph{approximation ratio} of a
$(\rho(\delta),k)$-multicoloring.  Because in a one-shot algorithm
(cf.\ the next section for a formal definition), a node of degree
$\delta$ cannot distinguish $G$ from $K_{\delta+1}$, the approximation
ratio of every one-shot algorithm needs to be at most $1$.

The multicoloring problem is related to the fractional coloring
problem in the following way. Assume that every node is assigned the
same number $c$ of colors and that the total number of colors is $k$.
Taking every color with fraction $1/c$ then leads to a fractional
$(k/c)$-coloring of $G$. Hence, in this case, $k/c$ is lower bounded
by the fractional chromatic number $\chi_f(G)$ of $G$. 

\subsection{One-Shot Algorithms}
\label{sec:model}

As outlined in the introduction, we are interested in local algorithms
to compute multicolorings of an $n$-node graph $G=(V,E)$.  For a
parameter $N\ge n$, we assume that every node $v$ has a unique ID
$x_v\in[N]$. In deterministic algorithms, every node has to compute a
color set based on its own ID as well as the IDs of its neighbors.
For randomized algorithms, we assume that nodes also know the random
bits of their neighbors. Formally, a one-shot algorithm can be
defined as follows.

\begin{definition}[One-Shot Algorithm] We call a distributed algorithm
  a one-shot algorithm if every node $v$ performs (a subset of) the
  following three steps:\\
  \hspace*{3mm}1. Generate sequence $R_v$ of random bits
  (deterministic algorithms: $R_v=\emptyset$)\\
  \hspace*{3mm}2. Send $x_v,R_v$ to all neighbors\\
  \hspace*{3mm}3. Compute solution based on $x_v$, $R_v$, and the
  received information
\end{definition}

Assume that $G$ is a network graph such that two nodes $u$ and $v$ can
directly communicate with each other iff they are connected by an edge
in $G$. In the standard \emph{synchronous message passing} model, time
is divided into rounds and in every round, every node of $G$ can send
a message to each of its neighbors. One-shot algorithms then exactly
correspond to computations that can be carried out in a single
communication round.


For deterministic one-shot algorithms, the output of every node $v$
is a function of $v$'s ID $x_v$ and the IDs of $v$'s neighbors. We
call this information on which $v$ bases its decisions, the
\emph{one-hop view} of $v$.

\begin{definition}[One-Hop View]
  Consider a node $v$ with ID $x_v$ and let $\Gamma_v$ be the set of
  IDs of the neighbors of $v$. We call the pair $(x_v,\Gamma_v)$ the
  one-hop view of $v$.
\end{definition}

Let $(x_u,\Gamma_u)$ and $(x_v,\Gamma_v)$ be the one-hop views of two
adjacent nodes. Because $u$ and $v$ are neighbors, we have
$x_u\in\Gamma_v$ and that $x_v\in\Gamma_u$. It is also not hard to see
that
\begin{equation}\label{eq:ngraphedge}
  \forall x_u,x_v\in[N]\text{ and }
  \forall \Gamma_u,\Gamma_v\in2^{[N]}\text{ such that }
  x_u\not=x_v, x_u\in\Gamma_v\setminus\Gamma_u, 
  x_v\in\Gamma_u\setminus\Gamma_v,
\end{equation}
there is a labeled graph that has two adjacent nodes $u$ and $v$ with
one-hop views $(x_u,\Gamma_u)$ and $(x_v,\Gamma_v)$, respectively.
Assume that we are given a graph with maximum degree $\Delta$ (i.e.,
for all one-hop views $(x_v,\Gamma_v)$, we have
$|\Gamma_v|\le\Delta$).  A one-shot vertex coloring algorithm maps
every possible one-hop view to a color. A correct coloring algorithm
must assign different colors to two one-hop views $(x_u,\Gamma_u)$ and
$(x_v,\Gamma_v)$ iff they satisfy Condition \eqref{eq:ngraphedge}.
This leads to the definition of the \emph{neighborhood graph}
$\mathcal{N}_1(N,\Delta)$ \cite{oneround} (the general notion of
neighborhood graphs has been introduced in \cite{linial92}). The nodes
of $\mathcal{N}_1(N,\Delta)$ are all one-hop views $(x_v,\Gamma_v)$
with $|\Gamma_v|\le\Delta$. There is an edge between $(x_u,\Gamma_u)$
and $(x_v,\Gamma_v)$ iff the one-hop views satisfy Condition
\eqref{eq:ngraphedge}. Hence, a one-shot coloring algorithm must
assign different colors to two one-hop views iff they are neighbors in
$\mathcal{N}_1(N,\Delta)$. The number of colors that are needed to
properly color graphs with maximum degree $\Delta$ by a one-shot
algorithm therefore exactly equals the chromatic number
$\chi\big(\mathcal{N}_1(N,\Delta)\big)$ of the neighborhood graph (see
\cite{oneround,linial92} for more details). Similarly, a one-shot
$(\rho(\delta),k)$-multicoloring algorithm corresponds to a
$(\rho(\delta),k)$-multicoloring of the neighborhood graph.

\section{Upper Bounds}
\label{sec:upperbounds}

In this section, we prove all the upper bounds claimed in Section
\ref{sec:intro}. We first prove that an efficient deterministic
one-shot multicoloring algorithm exists in Section
\ref{sec:existence}. Based on similar ideas, we derive an almost
optimal randomized algorithm in Section \ref{sec:randalg}. Finally, in
Section \ref{sec:constructive}, we introduce constructive methods to
obtain one-shot multicoloring algorithms. For all algorithms, we
assume that the nodes know the size of the ID space $N$ as well as
$\Delta$, an upper bound on the largest degree in the network. It
certainly makes sense that nodes are aware of the used ID space. Note
that it is straight-forward to see that there cannot be a non-trivial
solution to the one-shot multicoloring problem if the nodes do not
have an upper bound on the maximum degree in the network.

\subsection{Existence of an Efficient Deterministic Algorithm}
\label{sec:existence}


The existence of an efficient, deterministic one-shot multicoloring
algorithm is established by the following theorem.

\begin{theorem}\label{thm:probabilistic}
  Assume that we are given a graph with maximum degree $\Delta$ and
  node IDs in $[N]$. Then, for all $0<\eps\le1$, there is a
  deterministic, one-shot
  $\left(1-\eps,\Oh(\Delta^2\log(N)/\eps^2)\right)$-multicoloring
  algorithm.
\end{theorem}
\begin{proof}
  We use permutations to construct colors as described in
  \cite{oneround}.  For $i=1,\dots,k$, let $\prec_i$ be a global order
  on the ID set $[N]$.  A node $v$ with $1$-hop view $(x_v,\Gamma_v)$
  includes color $i$ in its color set iff $\forall y\in\Gamma_v:
  x_v\prec_i y$. It is clear that with this approach the color sets of
  adjacent nodes are disjoint. In order to show that nodes of degree
  $\delta$ obtain a $\rho/(\delta+1)$-fraction of all colors, we need
  to show that for all $\delta\in[\Delta]$, all $x\in[N]$, and all
  $\Gamma\in{[N]\setminus\{x\}\choose\delta}$, for all $y\in\Gamma$,
  $x\prec_i y$ for at least $k\rho/(\delta+1)$ global orders
  $\prec_i$. We use the probabilistic method to show that a set of
  size $k=2(\Delta+1)^2\ln(N)/\eps^2$ of global orders $\prec_i$
  exists such that every node of degree $\delta\in[\Delta]$ gets at
  least an $(1-\eps)/(\delta+1)$-fraction of the $k$ colors. Such a
  set implies that there exists an algorithm that satisfies the
  claimed bounds for all graphs with maximum degree $\Delta$ and IDs
  in $[N]$.

  Let $\prec_1,\dots,\prec_k$ be $k$ global orders chosen
  independently and uniformly at random. The probability that a node
  $v$ with degree $\delta$ and $1$-hop view $(x_v,\Gamma_v)$ gets
  color $i$ is $1/(\delta+1)$ (note that $|\Gamma_v|=\delta$). Let
  $X_v$ be the number of colors that $v$ gets. We have
  $\expect[X_v]=k/(\delta+1)\ge k/(\Delta+1)$. Using a Chernoff bound,
  we then obtain
  \begin{equation}\label{eq:existencechernoff}
  \pr\left[X_v<(1-\eps)\cdot\frac{k}{\delta+1}\right] =
  \pr\left[X_v<(1-\eps)\cdot\expect[X_v]\right] <
  e^{-\eps^2\expect[X_v]/2}\le
   \frac{1}{N^{\Delta+1}}.
  \end{equation}
  The total number of different possible one-hop views can be bounded as
  \(
  |\mathcal{N}_1(N,\Delta)| = 
  N\cdot\sum_{\delta=1}^\Delta{N-1\choose\delta} <
  N^{\Delta+1}.
  \)
  By a union bound argument, we therefore get that with positive
  probability, for all $\delta\in[\Delta]$, all possible one-hop
  views $(x_v,\Gamma_v)$ with $|\Gamma_v|=\delta$ get at least
  $(1-\eps)\cdot k/(\delta+1)$ colors. Hence, there exists a set of
  $k$ global orders on the ID set $[N]$ such that all one-hop views
  obtain at least the required number of colors.
\end{proof}

\noindent\textbf{Remark:}
Note that if we increase the number of permutations (i.e., the number
of colors) by a constant factor, all possible one-hop views $(x,\Gamma)$
with $|\Gamma|=\delta$ get a $(1-\eps)/(\delta+1)$-fraction of all
colors w.h.p.

\subsection{Randomized Algorithms}
\label{sec:randalg}

We will now show that with the use of randomization, the upper bound
of Section \ref{sec:existence} can be significantly improved if the
algorithm only needs to be correct w.h.p. We will again use random
permutations.  The problem of the deterministic algorithm is that the
algorithm needs to assign a large set of colors to all roughly
$N^\Delta$ possible one-hop views. With the use of randomization, we
essentially only have to assign colors to $n$ randomly chosen one-hop
views.

For simplicity, we assume that every node knows the number of nodes
$n$ (knowing an upper bound on $n$ is sufficient). For an integer
parameter $k>0$, every $v\in V$ chooses $k$ independent random numbers
$x_{v,1},\dots,x_{v,k}\in[kn^4]$ and sends these random numbers to all
neighbors. We use these random numbers to induce $k$ random
permutations on the nodes. Let $\Gamma(v)$ be the set of neighbors of
a node $v$. A node $v$ selects all colors $i$ for which
$x_{v,i}<x_{u,i}$ for all $u\in\Gamma(v)$.

\begin{theorem}\label{thm:randalg}
  Choosing $k=6(\Delta+1)\ln(n)/\eps^2$ leads to a randomized one-shot
  algorithm that computes a $(1-\eps,k)$-multicoloring w.h.p.
\end{theorem}


\noindent\textbf{Remark:}
In the above algorithm, every node has to generate
$\Oh(\Delta\log^2(n)/\eps^2)$ random bits and send these bits to the
neighbors. Using a (non-trivial) probabilistic argument, it is
possible to show that the same result can be achieved using only
$\Oh(\log n)$ random bits per node.


\subsection{Explicit Algorithms}
\label{sec:constructive}

We have shown in Section \ref{sec:existence} that there is a
deterministic one-shot algorithm that almost matches the lower bound
(cf.\ Theorem \ref{thm:detLB}). Unfortunately, the techniques of
Section \ref{sec:existence} do not yield an explicit algorithm. In
this section, we will present constructive methods to obtain a
one-shot multicoloring algorithm.

\begin{algorithm}[t]
  \caption{Explicit Deterministic Multicoloring Algorithm: 
    Basic Construction}
  \label{alg:constructivebasic}
  \begin{algorithmic}[1]
  \item[\textbf{Input:}] one-hop view $(x,\Gamma)$, parameter $\ell\ge0$
  \item[\textbf{Output:}] set $S$ of colors, initially $S=\emptyset$
    \FORALL{$(\alpha_0,\alpha_1,\ldots,\alpha_\ell)
      \in\field{q_0}\times\field{q_1}\times\cdots\times\field{q_\ell}$}
    \STATE $\beta_{0,x} := \varphi_{0,x}(\alpha_0)$;
    $\forall y\in\Gamma: \beta_{0,y} := \varphi_{0,y}(\alpha_0)$
    \FOR{$i:=1$ \textbf{to} $\ell$}
    \STATE $\beta_{i,x} := \varphi_{i,\beta_{i-1,x}}(\alpha_i)$;
    $\forall y\in\Gamma:
    \beta_{i,y} := \varphi_{i,\beta_{i-1,y}}(\alpha_i)$
    \ENDFOR
    \ENDFOR
    \IF{$\forall y\in\Gamma: \beta_{\ell,x}\not=\beta_{\ell,y}$}
    \STATE $S:=S\cup(\alpha_0,\alpha_1,\ldots,\alpha_\ell,\beta_{\ell,x})$
    \ENDIF
  \end{algorithmic}
\end{algorithm}

We develop the algorithm in two steps. First, we construct a
multicoloring where in the worst case, every node $v$ obtains the same
fraction of colors independent of $v$'s degree. We then show how to
increase the fraction of colors assigned to low-degree nodes. For an
integer parameter $\ell\ge0$, let $q_0,\ldots,q_\ell$ be prime powers
and let $d_0,\ldots,d_\ell$ be positive integers such that
$q_0^{d_0+1}\ge N$ and $q_i^{d_i+1}\ge q_{i-1}$ for $i\ge1$. For a
prime power $q$ and a positive integer $d$, let $\mathcal{P}(q,d)$ be
the set of all $q^{d+1}$ polynomials of degree at most $d$ in
$\field{q}[z]$, where $\field{q}$ is the finite field of order $q$. We
assume that that we are given an injection $\varphi_0$ from the ID set
$[N]$ to the polynomials in $\mathcal{P}(q_0,d_0)$ and injections
$\varphi_i$ from $\field{q_{i-1}}$ to $\mathcal{P}(q_i,d_i)$ for
$i\ge1$. For a value $x$ in the respective domain, let $\varphi_{i,x}$
be the polynomial assigned to $x$ by injection $\varphi_i$. The first
part of the algorithm is an adaptation of a technique used in a
coloring algorithm described in \cite{linial92} that is based on an
algebraic construction of \cite{erdos85}. There, a node $v$ with one-hop
view $(x,\Gamma)$ selects a color
$\big(\alpha,\varphi_{0,x}(\alpha)\big)$, where $\alpha\in\field{q_0}$
is a value for which $\varphi_{0,x}(\alpha)\not=\varphi_{0,y}(\alpha)$
for all $y\in\Gamma$ (we have to set $q_0$ and $d_0$ such that this is
always possible). We make two modifications to this basic algorithm.
Instead of only selecting one value $\alpha\in\field{q_0}$ such that
$\forall y\in\Gamma:\varphi_{0,x}(\alpha)\not=\varphi_{0,y}(\alpha)$,
we select all values $\alpha$ for which this is true. We then use
these values recursively (as if $\varphi_{i,x}(\alpha_i)$ was the ID
of $v$) $\ell$ times to reduce the dependence of the approximation
ratio of the coloring on $N$. The details of the first step of the
algorithm are given by Algorithm \ref{alg:constructivebasic}.

\begin{lemma}\label{lemma:basicalganalysis}
  Assume that for $0\le i\le\ell$, $q_i\ge f_i\Delta d_i$ where
  $f_i>1$. Then, Algorithm \ref{alg:constructivebasic} constructs a
  multicoloring with $q_\ell\cdot\prod_{i=0}^\ell q_i$ colors where
  every node at least receives a $\lambda/q_\ell$-fraction of all
  colors where $\lambda=\prod_{i=0}^\ell (1-1/f_i)$.
\end{lemma}
\begin{proof}
  All colors that are added to the color set in line 6 are from
  $\field{q_0}\times\field{q_1} \times\cdots\times
  \field{q_{\ell}}\times\field{q_{\ell}}$.  It is therefore clear that
  the number of different colors is $q_\ell\cdot\prod_{i=0}^\ell q_i$
  as claimed. From the condition in line 5, it also follows that the
  color sets of adjacent nodes are disjoint.

  To determine the approximation ratio, we count the number of colors,
  a node $v$ with one-hop view $(x,\Gamma)$ gets. First note that the
  condition in line 5 of the algorithm implies that (and is therefore
  equivalent to demand that) $\beta_{i,x}\not=\beta_{i,y}$ for all
  $y\in\Gamma$ and for all $i\in\{0,\dots,\ell\}$ because
  $\beta_{i,x}=\beta_{i,y}$ implies $\beta_{j,x}=\beta_{j,y}$ for all
  $j\ge i$. We therefore need to count the number of
  $(\alpha_0,\ldots,\alpha_{\ell})\in\field{q_0} \times
  \cdots\times\field{q_{\ell}}$ for which
  $\beta_{i,x}\not=\beta_{i,y}$ for all $i\in\{0,\dots,\ell\}$ and all
  $y\in\Gamma$. We prove by induction on $i$ that for $i<\ell$, there
  are at least $\prod_{j=0}^i q_j\cdot(1-1/f_j)$ tuples
  $(\alpha_0,\ldots,\alpha_i)\in\field{q_0}\times\cdots\field{q_i}$
  with $\beta_{j,x}\not=\beta_{j,y}$ for all $j\le i$. Let us first
  prove the statement for $i=0$. Because the IDs of adjacent nodes are
  different, we know that $\varphi_{0,x}\not=\varphi_{0,y}$ for all
  $y\in\Gamma$. Two different degree $d_0$ polynomials can be equal at
  at most $d_0$ values. Hence, for every $y\in\Gamma$,
  $\varphi_{0,x}(\alpha)=\varphi_{0,y}(\alpha)$ for at most $d_0$
  values $\alpha$. Thus, since $|\Gamma|\le\Delta$, there are at least
  $q_0-\Delta d_0\ge q_0\cdot(1-1/f_0)$ values $\alpha$ for which
  $\varphi_{0,x}\not=\varphi_{0,y}$ for all $y\in\Gamma$. This
  establishes the statement for $i=0$. For $i>0$, the argument is
  analogous. Let $(\alpha_0,\ldots,\alpha_{i-1})\in\field{q_0}
  \times\cdots\times\field{q_{i-1}}$ be such that
  $\beta_{j,x}\not=\beta_{j,y}$ for all $y\in\Gamma$ and all $j<i$.
  Because $\beta_{i-1,x}\not=\beta_{i-1,y}$, we have
  $\varphi_{i,x}\not=\varphi_{i,y}$. Thus, with the same argument as
  for $i=0$, there are at least $q_i\cdot(1-1/f_i)$ values $\alpha_i$
  such that $\beta_{i,x}\not=\beta_{i,y}$ for all $y\in\Gamma$.
  Therefore, the number of colors in the color set of every node is at
  least $\prod_{i=0}^{\ell} q_i\cdot\big(1-1/f_i\big) =
  \lambda\cdot\prod_{i=0}^\ell q_i$.  This is a
  $(\lambda/q_\ell)$-fraction of all colors.
\end{proof}

The next lemma specifies how the values of $q_i$, $d_i$, and $f_i$ can
be chosen to obtain an efficient algorithm.

\begin{lemma}\label{lemma:basicalgparams}
  Let $\ell$ be such that $\ln^{(\ell)}N>\max\{e,\Delta\}$. For $0\le
  i\le\ell$, we can then choose $q_i$, $d_i$, and $f_i$ such that
  Algorithm \ref{alg:constructivebasic} computes a multicoloring with
  $\Oh(\ell\Delta)^{\ell+2}\cdot\log_\Delta
  N\cdot\log_\Delta\ln^{(\ell)}N$ colors and such that every node gets
  at least a $1 / \big(4e^{9/4} \Delta
  \big\lceil\log_\Delta\ln^{(\ell)}N\big\rceil \big)$-fraction of all
  colors.
\end{lemma}

The number of colors that Algorithm \ref{alg:constructivebasic}
assigns to nodes with degree almost $\Delta$ is close to optimal even
for small values of $\ell$. If we choose
$\ell=\Theta(\log^*N-\log^*\Delta)$, nodes of degree $\Theta(\Delta)$
even receive at least a $(d/\Delta)$-fraction of all colors for some
constant $d$. Because the number of colors assigned to a node $v$ is
independent of $v$'s degree, however, the coloring of Algorithm
\ref{alg:constructivebasic} is far from optimal for low-degree nodes.
In the following, we show how to improve the algorithm in this
respect.

\begin{algorithm}[t]
  \caption{Explicit Deterministic Multicoloring Algorithm: Small
    Number of Colors}
  \label{alg:constructive1}
  \begin{algorithmic}[1]
  \item[\textbf{Input:}] one-hop view $(x,\Gamma)$,
    instances $\mathcal{A}_{2^i,N}$ for
    $i\in\big[\lceil\log\Delta\rceil\big]$ of Algorithm
    \ref{alg:constructivebasic}, parameter $\eps\in[0,1]$
  \item[\textbf{Output:}] set $S$ of colors, initially $S=\emptyset$
    \FORALL{$i\in\big[\lceil\log\Delta\rceil\big]$} 
    \STATE $\omega_i:=\left\lceil\left(\Delta/2^{i-1}\right)^\eps\cdot
    \big|\mathcal{C}_{2^{\lceil\log\Delta\rceil},N}\big|/
    \big|\mathcal{C}_{2^i,N}\big|\right\rceil$
    \ENDFOR
    \FORALL{$i\in\big\{\lceil\log|\Gamma|\rceil,\ldots,
      \lceil\log\Delta\rceil\big\}$}
    \FORALL{$c\in\mathcal{C}_{2^i,N}[x,\Gamma]$}
    \STATE \textbf{for all} $j\in[\omega_i]$ \textbf{do} 
    $S:=S\cup(c,i,j)$
    \ENDFOR
    \ENDFOR
  \end{algorithmic}
\end{algorithm}

Let $\mathcal{A}_{\Delta,N}$ be an instance of Algorithm
\ref{alg:constructivebasic} for nodes with degree at most $\Delta$ and
let $\mathcal{C}_{\Delta,N}$ be the color set of
$\mathcal{A}_{\Delta,N}$. Further, for a one-hop view $(x,\Gamma)$,
let $\mathcal{C}_{\Delta,N}[x,\Gamma]$ be the colors assigned to
$(x,\Gamma)$ by Algorithm $\mathcal{A}_{\Delta,N}$. We run instances
$\mathcal{A}_{2^i,N}$ for all $i\in\big[\lceil\log\Delta\rceil\big]$.
A node $v$ with degree $\delta$ chooses the colors of all instances
for which $2^i\ge\delta$.  In order to achieve the desired trade-offs,
we introduce an integer weight $\omega$ for each color $c$, i.e.,
instead of adding color $c$, we add colors $(1,c),\ldots,(\omega,c)$.
The details are given by Algorithm \ref{alg:constructive1}. The
properties of Algorithm \ref{alg:constructive1} are summarized by the
next theorem. The straight-forward proof is omitted.

\begin{theorem}\label{thm:explicitalg}
  Assume that in the instances of Algorithm
  \ref{alg:constructivebasic}, the parameter $\ell$ is chosen such
  that for all $\Delta$, $\mathcal{A}_{\Delta,N}$ assigns at least a
  $f(N)/\Delta$-fraction of the colors to every node. Then, for a
  parameter $\eps\in[0,1]$, Algorithm \ref{alg:constructive1} computes
  a $\big(\Omega(f(N)\eps/\delta^{\eps}),
  \Oh(|\mathcal{C}_{2\Delta,N}|\cdot\Delta^\eps/\eps)\big)$-multicoloring.
\end{theorem}

\begin{corollary}
  Let $\eps\in[0,1]$ and $\ell\ge 0$ be a fixed constant in all
  used instances of Algorithm \ref{alg:constructivebasic}. Then,
  Algorithm \ref{alg:constructive1} computes an
  $\big(\eps/\Oh(\delta^{\eps}\log_\Delta\ln^{(\ell)}N),
  \Oh(\Delta^{\ell+2} \cdot \log_\Delta
  N\cdot\log_\Delta\ln^{(\ell)}N)\big)$-multicoloring. In particular,
  choosing $\ell=0$ leads to an $\big(\eps/\Oh(\delta^\eps\log_\Delta
  N),\Oh(\Delta^2\log_\Delta^2N)\big)$-multicoloring. Taking the
  maximum possible value for $\ell$ in all used instances of
  Algorithm \ref{alg:constructivebasic} yields an
  $\big(\eps/\Oh(\delta^\eps), \Delta^{\Oh(\log^*N-\log^*\Delta)}
  \cdot \log_\Delta N\big)$-multicoloring.
\end{corollary}



\section{Lower Bounds}
\label{sec:lowerbounds}

In this section, we give lower bounds on the number of colors required
for one-shot multicoloring algorithms. In fact, we even derive the
lower bounds for algorithms that need to assign only one color to
every node, i.e., the results even hold for standard coloring
algorithms.

It has been shown in \cite{oneround} that every deterministic
one-shot $c$-coloring algorithm $\mathcal{A}$ can be interpreted as a
set of $c$ antisymmetric relations on the ID set $[N]$. Assume that
$\mathcal{A}$ assigns a color from a set $C$ with $|C|=c$ to every
one-hop view $(x,\Gamma)$. For every color $\alpha\in C$, there is a
relation $\colrel{\alpha}$ such that for all $x,y\in[N]$
$x\not\colrel{\alpha}y \lor y\not\colrel{\alpha}x$.  Algorithm
$\mathcal{A}$ can assign color $\alpha\in C$ to a one-hop view
$(x,\Gamma)$ iff $\forall y\in\Gamma: x\colrel{\alpha}y$.

For $\alpha\in C$, let
$\badcol{\alpha}{x}:=\{y\in[N]:x\not\colrel{\alpha}y\}$ be the set of
IDs that must not be adjacent to an $\alpha$-colored node with ID $x$.
To show that there is no deterministic, one-shot $c$-coloring
algorithm, we need to show that for every $c$ antisymmetric relations
$\colrel{\alpha_1},\ldots,\colrel{\alpha_c}$ on $[N]$, there is a
one-hop view $(x,\Gamma)$ such that $\forall i\in[c]:
\Gamma\cap\badcol{\alpha_i}{x}\not=\emptyset$. The following lemma is
a generalization of Lemma 4.5 in \cite{oneround} and key for the
deterministic and the randomized lower bounds. As the proof is along
the same lines as the proof of Lemma 4.5 in \cite{oneround}, it is
omitted here.




\begin{lemma}\label{lemma:badcol}
  Let $X\subseteq[N]$ be a set of IDs and let $t_1,\ldots,t_{\ell}$
  and $k_1,\ldots,k_{\ell}$ be positive integers such that
  $$
  t_i\cdot\big(\lambda(|X|-c)t_i-c\big)>2c(k_i-1)\quad
  \text{for }1\le i\le \ell
  \text{ and a parameter }\lambda\in[0,1].
  $$
  Then there exists an ID set $X'\subseteq X$ with $|X'|>(1-\ell\cdot
  \lambda)\cdot(|X|-c)$ such that for all $i\in[\ell]$,
  \begin{equation*}
  \forall x\in X', \forall \alpha_1,\ldots, \alpha_{t_i}\in C: 
  \sum_{j=1}^{t_i}\big|\badcol{\alpha_j}{x}\cap X\big|\ge k_i,
  \quad
  \forall x\in X', \forall \alpha\in C: 
  \badcol{\alpha}{x}\cap X\not=\emptyset.
  \end{equation*}
\end{lemma}

Based on several applications of Lemma \ref{lemma:badcol} (and based
on an $\Omega(\log\log N)$ lower bound in \cite{linial92}), it is
possible to derive an almost tight lower bound for deterministic
one-shot coloring algorithms. Due to lack of space, we only state the
result here.

\begin{theorem}\label{thm:detLB}
  If $N=\Omega(\Delta^2\log\Delta)$, every deterministic one-shot
  coloring algorithm needs at least $\Omega(\Delta^2+\log\log N)$
  colors.
\end{theorem}

\subsection{Randomized Lower Bound}
\label{sec:randomLB}

To obtain a lower bound for randomized multicoloring algorithms, we
can again use the tools derived for the deterministic lower bound by
applying Yao's principle. On a worst-case input, the best randomized
algorithm cannot perform better than the best deterministic algorithm
for a given random input distribution.  Choosing the node labeling at
random allows to again only consider deterministic algorithms.

We assume that the $n$ nodes are assigned a random permutation of the
labels $1,\ldots,n$ (i.e., every label occurs exactly once). Note that
because we want to prove a lower bound, assuming the most restricted
possible ID space makes the bound stronger. For an ID $x\in[n]$, we
sort all colors $\alpha\in C$ by increasing values of
$|\badcol{\alpha}{x}|$ and let $\alpha_{x,i}$ be the $i^\mathit{th}$
color in this sorted order. Further, for $x\in[n]$, we define
$b_{x,i}:=\big|\badcol{\alpha_{x,i}}{x}\big|$. In the following, we
assume that
\begin{equation}
  \label{eq:nofcolors}
  c=\kappa\cdot\frac{\Delta\lfloor\ln n\rfloor}{\lceil\ln\ln n\rceil+2}
  \quad\mbox{and}\quad
  n\ge12
  \quad\mbox{and}\quad
  n\ge\Delta\cdot\ln n
\end{equation}
for a constant $0<\kappa\le 1$ that will be determined later. By
applying Lemma \ref{lemma:badcol} in different ways, the next lemma
gives lower bounds on the values of $b_{x,i}$ for $n/2$ IDs $x\in[n]$.
\begin{lemma}\label{lemma:badcolrand}
  Assume that $c$ and $n$ are as given by Equation
  \eqref{eq:nofcolors} and let $0<\rho<1/3$ be a positive constant.
  Further, let $\tilde{t}=\big\lceil\rho\ln n/\ln\ln n\big\rceil$ and
  $t_i=2^{i-1}\cdot\lfloor\ln n\rfloor$ for $1\le i\le\ell$ where
  $\ell=\lceil\ln\ln n\rceil+2$. Then, for at least $n/2$ of all IDs
  $x\in[n]$, we have
  \[
  b_{x,1}\ge \frac{\ln\ln n}{44\kappa\cdot\ln n}\cdot\frac{n}{\Delta}
  - 1,
  \quad
  b_{x,\tilde{t}}\ge\frac{\rho}{48\kappa}\cdot\frac{n}{\Delta}
  - \frac{1}{2},
  \quad
  b_{x,t_i}\ge 2^{i-1}\cdot\left(\frac{1}{8\kappa}\cdot\frac{n}{\Delta}
  - \frac{1}{2}\right)\ \ \ 
  \mbox{for}\ \ \ 1\le i\le \ell.
  \]
\end{lemma}

In order to prove the lower bound, we want to show that for a randomly
chosen one-hop view $(x,\Gamma)$ with $|\Gamma|=\Delta$, the
probability that there is a color $\alpha\in C$ for which
$\Gamma\cap\badcol{\alpha}{x}=\emptyset$ is sufficiently small.
Instead of directly looking at random one-hop views $(x,\Gamma)$ with
$|\Gamma|=\Delta$, we first look at one-hop views with
$|\Gamma|\approx\Delta/e$ that are constructed as follows. Let
$X\subseteq[n]$ be the set of IDs $x$ of size $|X|\ge n/2$ for which
the bounds of Lemma \ref{lemma:badcolrand} hold. We choose $x_R$
uniformly at random from $X$. The remaining $n-1$ IDs are
independently added to a set $\Gamma_R$ with probability
$p=\frac{\Delta}{en}$. For a color $\alpha\in C$, let
$\mathcal{E}_\alpha$ be the event that
$\Gamma_R\cap\badcol{\alpha}{x_R}\not=\emptyset$, i.e.,
$\mathcal{E}_\alpha$ is the event that color $\alpha$ cannot be
assigend to the randomly chosen one-hop view $(x_R,\Gamma_R)$.

\begin{lemma}\label{lemma:fkg}
  The probability that the randomly chosen one-hop view cannot be
  assigned one of the $c$ colors in $C$ is bounded by
  \[
  \pr\left[\bigcap_{\alpha\in C}\mathcal{E}_\alpha\right]\ \ge\
  \prod_{\alpha\in C}\pr\big[\mathcal{E}_\alpha\big]\ \ge\
  \prod_{\alpha\in C}\left(
    1-e^{-\frac{\Delta}{en}\cdot|\badcol{\alpha}{x_R}|}
  \right)\ =\ 
  \prod_{i=1}^c\left(
    1-e^{-\frac{\Delta\cdot b_{x_R,i}}{en}}
  \right).
  \]
\end{lemma}
\begin{proof}
  Note first that for $\alpha\in C$, we have
  \[
  \pr\big[\overline{\mathcal{E}_\alpha}\big]\ =\ 
  \pr\big[\Gamma_R\cap\badcol{\alpha}{x_R}=\emptyset\big]\ =\
  (1-p)^{|\badcol{\alpha}{x_R}|}\ \le\ 
  e^{-p|\badcol{\alpha}{x_R}|}\ =\
  e^{-\frac{\Delta}{en}\cdot|\badcol{\alpha}{x_R}|}.
  \]
  It therefore remains to prove that the probability that all events
  $\mathcal{E}_\alpha$ occur can be lower bounded by the probability
  that would result for independent events. Let us denote the colors
  in $C$ by $\alpha_1,\ldots,\alpha_c$. We then have
  \begin{equation}\label{eq:fkg}
  \pr\left[\bigcap_{\alpha\in C}\mathcal{E}_\alpha\right]
  \ = \ \prod_{i=1}^c
  \pr\left[\mathcal{E}_{\alpha_i}\Bigg|\bigcap_{j=1}^{i-1}
    \mathcal{E}_{\alpha_j}\right]
  \ \ge\ 
  \prod_{i=1}^c\pr\big[\mathcal{E}_{\alpha_i}\big].
  \end{equation}
  The inequality holds because the events $\mathcal{E}_{\alpha}$ are
  positively correlated. Knowing that an element from a set
  $\badcol{\alpha}{x_R}$ is in $\Gamma_R$ cannot decrease the
  probability that an element from a set $\badcol{\alpha'}{x_R}$ is in
  $\Gamma_R$. Note that this is only true because the IDs are
  independently added to $\Gamma_R$. More formally, Inequality
  \eqref{eq:fkg} can also directly be followed from the FKG inequality
  \cite{fkg}.
\end{proof}

For space reasons, the following two lemmas are given without
proof.

\begin{lemma}\label{lemma:probbound}
  Assume that $c$ and $n$ are given as in \eqref{eq:nofcolors} where
  the constant $\kappa$ is chosen sufficiently small and let $\rho>0$
  be a constant as in Lemma \ref{lemma:badcolrand}. There is a
  constant $n_0>0$ such that for $n\ge n_0$,
  \( 
  \pr\left[\bigcap_{\alpha\in C}\mathcal{E}_\alpha\right]
  \ >\ \frac{1}{2n^{3\rho}}.
  \)
\end{lemma}

\begin{lemma}\label{lemma:randomview}
  Let  $(x,\Gamma)$ be a one-hop  view chosen uniformly at random from
  all one-hop views  with $|\Gamma|=\Delta$. If $\Delta\ge e(\ln n+2)$
  and $n$, $c$, and $\rho$ are as before, the probability that none of
  the  $c$ colors   can be   assigned to   $(x,\Gamma)$  is  at  least
  $1/(8n^{3\rho})$.
\end{lemma}

In the following, we call a node $u$ together with $\Delta$ neighbors
$v_1,\dots,v_\Delta$, a $\Delta$-star.

\begin{theorem}\label{thm:randLB}
  Let $G$ be a graph with $n$ nodes and $2n^\eps$ disjoint
  $\Delta$-stars for a constant $\eps>0$. On $G$, every randomized
  one-shot coloring algorithm needs at least $\Omega(\Delta\log
  n/\log\log n)$ colors in expectation and with high probability.
\end{theorem}
\begin{proof}
  W.l.o.g., we can certainly assume that $n\ge n_0$ for a sufficiently
  large constant $n_0$. We choose $\rho\le\eps/4$ and consider $n^\eps$
  of the $2n^\eps$ disjoint $\Delta$-stars. Let us call these $n^\eps$
  $\Delta$-stars $S_1,\dots,S_{n^\eps}$. Assume that the ID assignment
  of the $n$ nodes of $G$ is chosen uniformly at random from all ID
  assignments with IDs $1,\dots,n$. The IDs of the star $S_1$ are
  perfectly random. We can therefore directly apply Lemma
  \ref{lemma:randomview} and obtain that the probability that the
  center node of $S_1$ gets no color is at least $1/(8n^{3\rho})$.
  Consider star $S_2$. The IDs of the nodes of $S_2$ are chosen at
  random among the $n-\Delta-1$ IDs that are not assigned to the nodes
  of $S_1$. Applying Lemma \ref{lemma:randomview} we get that the
  probability that $S_2$ does not get a color is at least
  $1/(8(n-\Delta-1)^{3\rho})\ge1/(8n^{3\rho})$ independently of
  whether $S_1$ does get a color. The probability that the starts
  $S_1,\dots,S_{n^\eps}$ all get a color therefore is at most
  \[
  \prod_{i=0}^{n^\eps-1}\left(1-\frac{1}{8(n-i(\Delta+1))^{3\rho}}\right)
  \le
  \left(1-\frac{1}{8n^{3\rho}}\right)^{n^\eps}
  \le
  e^{-\frac{n^\eps}{8n^{3\rho}}}
  \le
  e^{-n^\rho/8}.
  \]
  Hence, there is a constant $\eta>0$ such that $\eta\Delta\ln
  n/\ln\ln n$ colors do not suffice with probability at least
  $1-e^{-n^\rho/8}$ for a positive constant $\rho$. The lemma thus
  follows.
\end{proof}




\end{document}